\documentclass[letterpaper,11pt]{article}
\usepackage{microtype}
\usepackage[letterpaper,margin=1in]{geometry}
\usepackage{graphicx}
\usepackage{authblk}
\usepackage{hyperref}
\usepackage{amsthm}
\usepackage{amsmath}
\usepackage{amssymb}
\usepackage[capitalize, nameinlink]{cleveref}
\usepackage[numbers,sort&compress]{natbib}
\usepackage{thmtools,thm-restate}
\usepackage{nicefrac}
\usepackage{xcolor,xspace}
\usepackage{algpseudocode}
\usepackage{algorithm}
\usepackage{caption}
\usepackage{mathtools}
\usepackage{outlines}
\usepackage{fancyhdr}
\usepackage{tcolorbox}
\usepackage{enumitem}
\usepackage{thm-restate}

\newcommand{\fA}{\mathcal A}

\newcommand{\Vs}{V_{\operatorname{semi}}}
\newcommand{\Es}{E_{\operatorname{semi}}}

\newcommand{\ndeg}{\operatorname{deg}}
\newcommand{\rank}{\operatorname{rank}}

\newcommand{\EdgelistP}{Edge-list variant of a node-edge-checkable problem}
\newcommand{\NodelistP}{Node-list variant of a node-edge-checkable problem}
\newcommand{\mset}{\psi}
\newcommand{\mcN}{\mathcal{N}}
\newcommand{\mcE}{\mathcal{E}}
\newcommand{\mcL}{\mathcal{L}}
\newcommand{\Otilde}{\widetilde{O}}
\newcommand{\hin}{h_{\operatorname{in}}}
\newcommand{\hout}{h_{\operatorname{out}}}
\newcommand{\edgedeg}{\operatorname{edge-degree}}
\renewcommand{\kappa}{10}

\DeclarePairedDelimiter\ceil{\lceil}{\rceil}

\newtheorem{theorem}{Theorem}
\newtheorem{lemma}[theorem]{Lemma}

\newtheorem*{property*}{Property}
\newtheorem{definition}[theorem]{Definition}

\date{}
\title{Towards Optimal Deterministic LOCAL Algorithms on Trees}
\author[1]{Sebastian Brandt}
\author[1]{Ananth Narayanan}
\affil[1]{CISPA Helmholtz Center for Information Security}

\begin{document}
\maketitle

\begin{abstract}
    While obtaining optimal algorithms for the most important problems in the LOCAL model has been one of the central goals in the area of distributed algorithms since its infancy, tight complexity bounds are elusive for many problems even when considering \emph{deterministic} complexities on \emph{trees}.
    We take a step towards remedying this issue by providing a way to relate the complexity of a problem $\Pi$ on trees to its truly local complexity, which is the (asymptotically) smallest function $f$ such that $\Pi$ can be solved in $O(f(\Delta) + \log^* n)$ rounds.
    More specifically, we develop a transformation that takes an algorithm $\mathcal A$ for $\Pi$ with a runtime of $O(f(\Delta) + \log^* n)$ rounds as input and transforms it into an $O(f(g(n)) + \log^* n)$-round algorithm $\fA'$ on trees, where $g$ is the function that satisfies $g(n)^{f(g(n))} = n$.
    If $f$ is the truly local complexity of $\Pi$ (i.e., if $\fA$ is asymptotically optimal), then $\fA'$ is an asymptotically optimal algorithm on trees, conditioned on a natural assumption on the nature of the worst-case instances of $\Pi$.
    
    Our transformation works for any member of a wide class of problems, including the most important symmetry-breaking problems.
    As an example of our transformation we obtain the first strongly sublogarithmic algorithm for $(\text{edge-degree + 1})$-edge coloring (and therefore also $(2 \Delta - 1)$-edge coloring) on trees, exhibiting a runtime of $O(\log^{12/13} n)$ rounds.
    This breaks through the $\Omega(\log n/\log \log n)$-barrier that is a fundamental lower bound for other symmetry-breaking problems such as maximal independent set or maximal matching (that already holds on trees), and proves a separation between these problems and the aforementioned edge coloring problems on trees.

    We extend a subset of our results to graphs of bounded arboricity, obtaining, for instance, the aforementioned upper bound for edge colorings also on planar graphs.
\end{abstract}
\pagenumbering{gobble}
\clearpage
\pagenumbering{arabic}

\section{Introduction}
Since its beginning in the 1980s~\cite{linial1987distributive}, one of the central objectives of research on distributed graph algorithms has been to improve the best known lower and upper bounds for the complexities of the most important problems in the field, coming closer and closer to the ultimate goal of obtaining (asymptotically) tight bounds.
Much of the research has focused on a setting that we will also consider in this work: the model of computation is the standard LOCAL model of distributed computing~\cite{linial1992locality,peleg2000}, the considered problem class is the class of \emph{locally checkable\footnote{Informally speaking, a locally checkable problem is a problem that can be defined via local constraints such that a global solution to the problem is correct if and only if the local constraints are satisfied around each node or edge. A simple example are proper coloring problems where the local constraints are of the form that the two endpoints of an edge have to have different colors. For formal definitions, see \Cref{sec:prelims}.} problems}~\cite{naor1995what,brandt2019automatic}, and complexities are studied as a function of the number $n$ of the nodes of the input graph.
For this setting, the last decade has seen a tremendous amount of progress towards closing the gap between lower and upper bounds: on the lower bound side, the round elimination technique~\cite{brandt2016lower,brandt2019automatic} provided a new, powerful tool for proving lower bounds and is responsible for most of the state-of-the-art lower bounds known for fundamental problems~\cite{brandt2019automatic,chang2019exponential,chang2019distributed,brandt2020truly,balliu2021lower,balliu2022distributed}, while, on the upper bound side, techniques based on network decompositions and rounding considerably advanced the state of the art (see, e.g., \cite{fischer2020improved,rozhovn2020polylogarithmic,ghaffari2021improved,ghaffari2022deterministic,faour2023local,ghaffari2023faster,ghaffari2023improved}) and yield the current best known complexity upper bounds for many problems\cite{ghaffari2024near}.

However, despite all of these efforts, to the best of our knowledge there are only two natural\footnote{We note that it is possible to artificially construct problems in a way that makes it easy to prove tight bounds for the obtained problem but we are not aware of any such problem that has been studied in its own right independently of this property.} problems for which nontrivial\footnote{We call a tight bound ``trivial'' if it is of size $\Theta(0)$, $\Theta(1)$, or $\Theta(n)$} tight bounds are known (on general graphs): sinkless orientation, for which a tight bound of $\Theta(\log n)$ is known~\cite{ghaffari2017distributed,chang2019exponential} and $(2,2)$-ruling edge set, for which a tight bound of $\Theta(\log^* n)$ is known~\cite{linial1992locality,kuhn2018deterministic}.
In particular, for none of the fundamental symmetry-breaking problems, tight bounds are known; for instance, the state-of-the-art bounds for MIS and maximal matching are $\Omega(\log n/\log \log n)$ rounds~\cite{balliu2021lower} and $\Otilde(\log^{5/3} n)$ rounds\footnote{We will use $\Otilde(\cdot)$ to hide factors logarithmic in the argument, i.e., in this case $(\log \log n)$-factors.}~\cite{ghaffari2024near}, while, for $(\Delta + 1)$-coloring and $(2\Delta - 1)$-edge coloring, they are $\Omega(\log^* n)$ rounds~\cite{linial1992locality} and $\Otilde(\log^{5/3} n)$ rounds~\cite{ghaffari2024near}.

Perhaps, this might not come as a surprise when considering the state of the art in the simpler setting of \emph{deterministic} complexities on \emph{trees}: even there, tight bounds are still elusive for many fundamental problems, including the aforementioned coloring problems.
Two exceptions are the maximal independent set (MIS) problem and the maximal matching problem for which tight bounds of $\Theta(\log n/\log \log n)$ rounds on trees are known~\cite{barenboim2010sublogarithmic,barenboim2013distributed,balliu2022distributed}.
 It stands to reason that we need to understand the situation on trees before we can have hope to obtain tight bounds on general graphs (in particular in light of the fact that all lower bounds obtained via round elimination also hold on trees), raising the following fundamental question.

\vspace{0.15cm}
\begin{tcolorbox}
	
	\textbf{Open Question}
    
    \noindent How can we obtain \emph{optimal} (deterministic) algorithms on trees?
\end{tcolorbox}

\paragraph{Truly local complexities.}
While understanding the complexity of a problem \emph{as a function of $n$} has been the main focus of research on distributed graph algorithms, a second research direction that has received a lot of attention is to understand the complexity of a problem \emph{as a function of the maximum degree $\Delta$ of the input graph}, where only a minimal additive dependency on $n$ is allowed.
More concretely, a large number of works~\cite{panconesi2001some,kuhn2006complexity,barenboim2014distributed,barenboim2016deterministic,fraigniaud2016local,balliu2020distributed,brandt2020truly,kuhn2020faster,maus2020local,balliu2021improved,balliu2021lower,balliu2022distributed,balliu2022edge,balliu2022ruling} have studied what is known as the \emph{truly local complexity} of a problem~\cite{maus2020local}:
a function $f: \mathbb{Z}_{\geq 0} \rightarrow \mathbb{Z}_{\geq 0}$ is called the \emph{truly local complexity} of a problem $\Pi$ if there exists some algorithm that solves $\Pi$ in $O(f(\Delta) + \log^* n)$ rounds but no algorithm that solves $\Pi$ in $o(f(\Delta)) + O(\log^* n)$ rounds.
The choice of the additive $\log^*$-dependency\footnote{The function $\log^*$ is defined as the minimum number of times the $\log$-function has to be applied recursively to the argument to obtain a value that is at most $1$.} on $n$ in the definition of the truly local complexity is due to the fact that for a large number of problems, including all of the fundamental symmetry-breaking problems, such a dependency is unavoidable due to~\cite{linial1992locality}, while for many problems an algorithm with runtime $O(f(\Delta) + \log^* n)$ rounds indeed exists for some function $f$.

As with complexities as a function of $n$, the main objective of research on the truly local complexity has been to come closer to obtaining tight bounds by improving lower and upper bounds.
While, for problems such as $(\Delta + 1)$-coloring and $(2 \Delta - 1)$-edge coloring, the truly local complexity is still wide open, tight bounds have been achieved for other problems, such as a bound of $\Theta(\Delta)$ for MIS~\cite{barenboim2014distributed,balliu2022distributed} and maximal matching~\cite{panconesi2001some,balliu2021lower}.
But even for the two coloring problems, there has been ample improvement regarding bounds on the truly local complexity~\cite{panconesi2001some,kuhn2006complexity,barenboim2014distributed,barenboim2016deterministic,fraigniaud2016local,kuhn2020faster,maus2020local,balliu2022edge}: while for both problems no superconstant lower bound on the truly local complexity is known, the currently best known upper bound for the truly local complexity of $(\Delta + 1)$-coloring stands at $O(\sqrt{\Delta \log \Delta})$ rounds~\cite{maus2020local}, and for $(2\Delta - 1)$-coloring even at $O(\log^{12} \Delta)$ rounds~\cite{balliu2022edge}.
We remark that these bounds hold also for $(\deg + 1)$-coloring and $(\text{edge-degree} + 1)$-edge coloring, respectively.
In general, it appears that making progress on bounds for the truly local complexity of a problem may be (perhaps considerably) easier than doing so for bounds for the complexity of the problem on trees (as a function of $n$).

In fact, the round elimination technique mentioned above as a lower bound tool can similarly be used to prove upper bounds (as elaborated upon, e.g., in~\cite[Section 1.2]{balliu2021improved}.
In either case, when considering a problem that can be solved in $O(f(\Delta) + \log^* n)$ rounds for some function $f$, the bounds that round elimination produces are inherently bounds \emph{on the truly local complexity}; any bounds obtained as a function of $n$ are obtained via further manipulation of these bounds.
As such there already exists a technique specifically for the task of proving bounds on the truly local complexity---unlike for the task of proving bounds as a function of $n$ on trees.

\subsection{Our Contributions}
As our main contribution, we develop a transformation that relates complexities as a function of $n$ on trees to truly local complexities.
More specifically, for any problem $\Pi$ from a large class of problems, our transformation takes an algorithm for $\Pi$ with a deterministic runtime guarantee of the form $O(f(\Delta) + \log^* n)$ rounds for some (monotonically non-decreasing, non-zero) function\footnote{When considering complexity functions, we will for simplicity assume that they are (monotonically non-decreasing) continuous functions from the space of nonnegative reals to the space of nonnegative reals that map $0$ to $0$ (as this simplifies matters, e.g., the definition of the function $g$, considerably). When it is necessary to transform a traditional complexity function going from $\mathbb{Z}_{\geq 0}$ to $\mathbb{Z}_{\geq 0}$ into this form, simply choose any extension to the nonnegative reals that does not violate the monotonicity (and leaves all already defined function values unchanged). Moreover, we call a function \emph{non-zero} if there is at least one positive real number at which the function has a positive value.} $f$ as input and returns a deterministic algorithm with a runtime guarantee of $O(f(g(n)))$ rounds on trees, where $g$ is the function satisfying $g(n)^{f(g(n))} = n$.\footnote{We remark that the fact that $f$ is continuous, monotone and non-zero (together with $f(0) = 0$) implies that $g$ exists and is unique.}

Our transformation is applicable to two classes of locally checkable problems, which, roughly speaking, can be described as follows.
The first class $\mathcal P_1$ contains all locally checkable\footnote{We already note that, for technical reasons it will be required that the problems are given in the so-called \emph{node-edge-checkability} formalism, which is presented formally in~\Cref{sec:prelims}.} node-labeling problems for which there exists a sequential algorithm that solves the problem---even when the input graph comes with a correct partial solution---by assigning an output label for each node (presented to the algorithm in an adversarial order) by only taking the $1$-hop neighborhood (including the outputs chosen so far for nodes in the neighborhood) into account.\footnote{For readers familiar with the SLOCAL model~\cite{ghaffari2017complexity}, informally speaking, this can be thought of as the class of problems for which there exists an algorithm with locality $1$ in the SLOCAL model that also works for a (suitably defined) list version of the problem.}
Amongst others, this class contains the fundamental problems of MIS, $(\Delta + 1)$-coloring, and $(\deg + 1)$-coloring.

Dually, the second class $\mathcal P_2$ contains all edge-labeling problems for which there exists a sequential algorithm that solves the problem---even when the input graph comes with a correct partial solution---by assigning an output label for each edge (presented to the algorithm in an adversarial order) by only taking the $1$-hop edge neighborhood (which includes all information associated with the nodes and edges incident and adjacent to the currently processed edge) into account.
Amongst others, this class contains the fundamental problems of maximal matching, $(2\Delta - 1)$-edge coloring, and $(\text{edge-degree + 1})$-edge coloring.

We remark that the precise definitions of the two problem classes are quite technical and contain more problems than those that are captured by the informal outline provided above.
Disregarding this formal imprecision for a little while longer, we can now phrase our main contribution informally in the form of two theorems.
The first covers the aforementioned node-labeling problems.

\begin{theorem}[Simplified version of~\Cref{thm:nodesveryformal}]\label{thm:nodesinformal}
    Let $\Pi$ be a problem from $\mathcal P_1$ and $f$ a monotonically non-decreasing, non-zero function such that $\Pi$ can be solved in $O(f(\Delta) + \log^* n)$ rounds.
    Then $\Pi$ can be solved in $O(f(g(n)) + \log^* n)$ rounds on trees, where $g$ is the function satisfying $g(n)^{f(g(n))} = n$.
\end{theorem}

The second theorem covers the aforementioned edge-labeling problems.
However, we in fact prove a stronger theorem for this class of problems that works\footnote{When considering graphs of arboricity $a$, we implicitly assume that $a$ is known to the nodes.} for any graph of arboricity at most $a$ (and in the case of trees, i.e., $a = 1$, provides the dual to \Cref{thm:nodesinformal}).
The \emph{arboricity $a$} of a graph $G = (V,E)$ is defined as the smallest number of forests with node set $V$ such that each edge in $E$ appears in precisely one of the forests. 

\begin{theorem}[Simplified version of~\Cref{thm:edgesveryformal}]\label{thm:edgesinformal}
    Let $\Pi$ be a problem from $\mathcal P_2$ and $f$ a monotonically non-decreasing, non-zero function such that $\Pi$ can be solved in $O(f(\Delta) + \log^* n)$ rounds.
    Then $\Pi$ can be solved in $O\big(a+\frac{f(g(n))}{1-\log_{g(n)}a}+\log^{*}n\big)$ rounds on graphs of arboricity at most $a\leq \frac{g(n)}{5}$, where $g$ is the function that satisfies $g(n)^{f(g(n))}=n$.
\end{theorem}

We note that both theorems are constructive in the sense that given any algorithm with a runtime of $O(f(\Delta) + \log^* n)$ rounds, the proofs of the theorems provide a way to transform it into an algorithm with a runtime of $O(f(g(n)) + \log^* n)$ rounds, resp.\ $O\big(a+\frac{f(g(n))}{1-\log_{g(n)}a}+\log^{*}n\big)$ rounds. 
In the following we put the two theorems into context and discuss their ramifications.

\paragraph{Tightness of the obtained bounds on trees.}
How close to the true tight (deterministic) complexity on trees will the bound achieved by our transformation be?
To shed some light on this question, let us take a look at the current status of known lower bounds in the LOCAL model.
To the best of our knowledge, all state-of-the art (deterministic) lower bounds for locally checkable problems (given in the node-edge-checkability formalism described in~\Cref{sec:prelims}) are achievable via round elimination and already hold on regular balanced\footnote{We call a tree regular and balanced if every non-leaf node of the tree has the same degree and there is a ``root'' node that has the same distance to each leaf. For each fixed $\Delta$, a $\Delta$-regular balanced tree with $n$ nodes exists for infinitely many $n$; if it is desired that a $\Delta$-regular balanced tree exists for \emph{each} positive integer $n$, the definition can be naturally adapted by allowing the nodes in the ``non-leaf layer'' furthest from the root to have arbitrary degrees between $0$ and $\Delta$, which does not affect the obtained lower bounds asymptotically.} trees.
More precisely, when the considered problem admits an algorithm with a runtime of the form $O(f(\Delta) + \log^* n)$ for some function $f$, then (as already briefly discussed above) the lower bound is essentially achieved by first obtaining a lower bound on the truly local complexity of the problem (that already holds on regular balanced trees), and then lifting the bound to a lower bound as a function of $n$ via a mechanical process (see, e.g.,
\cite{balliu2022distributed}).
If the lower bound on the truly local complexity is, say, $\Omega(h(\Delta))$, then this mechanical process first produces a lower bound of $\Omega(\min\{ h(\Delta), \log_{\Delta} n\})$ that holds for any $\Delta$, and then turns this bound into a lower bound purely as a function in $n$ by setting $h(\Delta)$ and $\log_{\Delta} n$ to be equal.
Observe that setting these two expression to be equal corresponds precisely to solving $\Delta^{h(\Delta)} = n$ for $\Delta$, and the lower bound is then achieved by inserting the obtained expression for $\Delta$ (which is a function in $n$) into $h(\Delta)$.
In other words (by using $g(n)$ to denote the aforementioned function in $n$ that expresses $\Delta$), the lower bound is precisely $\Omega(h(g(n)))$ where $g$ is the function satisfying $g(n)^{h(g(n))} = n$ (and this lower bound still holds on regular balanced trees).

By additionally observing that the above argumentation only requires that the lower bound on the truly local complexity holds already if the input tree is balanced and regular (irrespective of whether it is achieved via round elimination or not), we obtain the following: if $\Theta(f(\Delta))$ is the tight truly local complexity of a considered problem, then the upper bounds yielded by \Cref{thm:nodesinformal,thm:edgesinformal} are asymptotically tight, assuming that the lower bound on the truly local complexity already holds if the input tree is balanced and regular and that the problem cannot be solved in $o(\log^* n)$ rounds.
As this assumption is satisfied for all superconstant lower bounds for locally checkable problems in the literature we are aware of, it is plausible to assume that it holds for the vast majority of such problems.
For all of these problems, our theorems therefore reduce the task of obtaining tight bounds on trees to the task of proving tight bounds for the truly local complexity, which, as discussed above is a task that might be considerably easier to approach and comes with a ready-made tool in the form of round elimination.
As such, our work makes substantial progress in answering the open question stated above.

\paragraph{Concrete implications.}
Moreover, even when a tight complexity might not be in sight yet, our transformation provides a tool to translate improvements on the best known bounds for the truly local complexity to improvements for the complexity (as a function of $n$) on trees and graphs of bounded arboricity.
Concretely, we will make use of the recent breakthrough on the complexity of the $(\text{edge-degree} + 1)$-edge coloring problem by~\cite{balliu2022edge}, which yields an upper bound of $O(\log^{12} n)$ rounds\footnote{Concretely, the exponent of $12$ follows from~\cite[Theorem D.4, arXiv version]{balliu2022edge}.} to obtain the state-of-the-art complexity on trees.
By observing that $(\text{edge-degree} + 1)$-edge coloring is a problem from $\mathcal P_2$, we obtain the following result by~\Cref{thm:edgesinformal}.\footnote{We note that, for technical reasons, the precise bound on $a$ given in~\cref{thm:edgedegplusoneintro} requires the formal version \Cref{thm:edgesveryformal} of~\Cref{thm:edgesinformal} while only a slightly weaker bound on $a$ can be obtained using~\Cref{thm:edgesinformal}.}

\begin{restatable}{theorem}{edgecolthm}\label{thm:edgedegplusoneintro}
    The complexity of $(\text{edge-degree} + 1)$-edge coloring is $O(\log^{12/13} n)$ rounds on trees and
    $O\big(a + \log^{12/13} n\big)$ rounds on graphs of arboricity at most $a \leq 2^{\log^{1/13}n}$.
\end{restatable}

As $(2 \Delta - 1)$-edge coloring is a problem that is at most as hard as $(\text{edge-degree} + 1)$-edge coloring, we obtain the upper bounds of~\Cref{thm:edgedegplusoneintro} also for $(2 \Delta - 1)$-edge coloring.
Previously, for $(2 \Delta - 1)$-edge coloring, the best upper bound  known on trees was $O(\log n/\log \log n)$ and the best upper bound knows for graphs of arboricity $a$ was $O(\log n/\log \log n)$ if $a \leq \log^{1-\varepsilon} n$ for some constant $\varepsilon > 0$ and $O(a + \log n)$ for any $a$ \cite{barenboim2013distributed}.
The approach from~\cite{barenboim2013distributed} can also be applied to achieve the same bounds for $(\text{edge-degree} + 1)$-edge coloring.
\Cref{thm:edgedegplusoneintro} provides the first strongly sublogarithmic-time\footnote{A function (in $n$) is called strongly sublogarithmic if it is in $O(\log^{\beta} n)$ for some constant $\beta < 1$.} algorithm for one of the fundamental symmetry-breaking problems on trees.
In particular, it breaks through the $\Omega(\log n/\log \log n)$-barrier that constitutes a lower bound for other symmetry-breaking problems such as MIS or maximal matching on trees~\cite{balliu2021lower,balliu2022distributed}, and thereby provides a separation between the aforementioned edge coloring problems and MIS/maximal matching on trees. 

The arboricity result of \Cref{thm:edgedegplusoneintro} implies that $(\text{edge-degree} + 1)$-edge coloring and $(2\Delta - 1)$-edge coloring can be solved in strongly sublogarithmic time on all graphs with strongly sublogarithmic arboricity.
In particular, we obtain the upper bound of $O(\log^{12/13} n)$ rounds on all graphs with constant arboricity such as, e.g., planar graphs.

Our results can also be used to reprove the state-of-the-art upper bounds on trees for MIS and maximal matching~\cite{barenboim2010sublogarithmic,barenboim2013distributed} in a generic manner.
For the last remaining of the ``Big Four'' symmetry-breaking problems, $(\Delta + 1)$-coloring, recall that the truly local complexity is still wide open, with no superconstant lower bound known and the best known upper bound being $O(\sqrt{\Delta \log \Delta})$ rounds~\cite{maus2020local}.
On trees, the best bounds known are $O(\log n/\log \log n)$ rounds due to~\cite{barenboim2010sublogarithmic} and $\Omega(\log^* n)$ rounds due to~\cite{linial1992locality}.
We observe that, by \Cref{thm:nodesinformal}, any large enough improvement of the known upper bound for the truly local complexity of $(\Delta + 1)$-coloring would yield an improvement for the complexity of $(\Delta + 1)$-coloring on trees; e.g., improvements to $O(2^{\sqrt{\log \Delta}})$ or $O(\log^{5} \Delta)$ would yield upper bounds on trees of $O(\log n/\log^2 \log n)$ or $O(\log^{5/6} n)$, respectively.
This holds similarly for other problems such as $(\deg + 1)$-coloring or the aforementioned edge coloring problems; e.g., improving the exponent in the best known upper bound $O(\log^{12} \Delta)$ for the truly local complexity of $(2\Delta - 1)$-edge coloring to $1$ would yield an upper bound of $O(\sqrt{\log n})$ rounds for $(2\Delta - 1)$-edge coloring on trees.

\paragraph{Towards tight bounds on general graphs?}
A conjecture that is consistent with the current state of the art is that for all natural locally checkable problems the tight complexity \emph{on general graphs} is the same as on trees.
In fact, for many problems, the best known upper bounds have steadily improved over the last years, inching closer to the lower bounds that already hold on trees.
Assuming this plausible conjecture, it stands to reason that for the ultimate goal of obtaining tight bounds on general graphs, we need a fine-grained approach that differentiates between all problems that have different truly local complexity---such as, e.g., maximal matching and $(2\Delta - 1)$-edge coloring.
We see our work also as a step towards this goal by providing such a differentiation. 

\subsection{Our Approach in a Nutshell}
After establishing the essential idea of reducing the task of designing algorithms (and proving upper bounds) on trees and graphs of arboricity $a$ to the task of designing algorithms achieving a good truly local complexity, our approach for obtaining such a reduction is conceptually simple.
At the heart of our approach lies the idea of decomposing the input graph into two parts with different properties: The first part is simply required to have a maximum degree of $k$, for some suitably chosen parameter $k$.
In the case of~\Cref{thm:nodesinformal}, each connected component of the second part is required to have diameter $O(\log_k n)$, while in the case of~\Cref{thm:edgesinformal}, the second part is required to admit a decomposition into $O(a)$ parts for each of which each connected component must have constant diameter.
Our algorithm for solving a given problem that admits an algorithm with a runtime of $O(f(\Delta) + \log^* n)$ for some function $f$ then proceeds as follows: compute a decomposition as described above in $O(\log_{\frac{k}{a}} n + \log^* n)$ rounds, then solve the first part in $O(f(k) + \log^* n)$ rounds using the aforementioned algorithm, and finally, in the case of~\Cref{thm:nodesinformal}, solve the second part in $O(\log_k n)$ rounds, and, in the case of~\Cref{thm:edgesinformal}, solve each of the $O(a)$ parts of the second part one after the other in a constant number of rounds by simply gathering the whole part in one node, computing a solution for the whole part in that node, and then distributing the solution to the other nodes in that part.
By setting $k := g(n)$ (where $g$ is the function from \Cref{thm:nodesinformal,thm:edgesinformal}), the described approach computes the respectively used decomposition in $O(\log_{\frac{g(n)}{a}} n + \log^* n) \subseteq O(\frac{f(g(n))}{1 - \log_{g(n)} a} + \log^* n)$ rounds\footnote{For the respective calculation, see the proof of~\Cref{thm:edgesveryformal}.}, solves the first part in $O(f(g(n)) + \log^* n)$ rounds, and solves the second part in $O(\log_{g(n)} n) = O(f(g(n)))$ or $O(a)$ rounds, yielding an overall runtime of $O(a + \frac{f(g(n))}{1 - \log_{g(n)} a} + \log^* n)$ rounds, which simplifies to $O(f(g(n)) + \log^* n)$ in the case of trees (for which $a = 1$ holds).

We will use different decompositions for our approach, depending on whether we prove \Cref{thm:nodesinformal} or \Cref{thm:edgesinformal}. 
For \Cref{thm:nodesinformal}, our decomposition heavily relies on a $O(\log_k n)$-round rake-and-compress~\cite{miller1985parallel} process (parameterized by some parameter $p$) introduced by~\cite{chang2019distributed} that guarantees that when a node $v$ is removed during the process, the degree of $v$ is $1$ or the degree of every node in $v$'s $1$-hop neighborhood is at most $p$.
As we show, when executed with parameter $p = k = g(n)$, the rake-and-compress decomposition obtained from this process can be transformed into a decomposition into two parts as described above in a simple way.

For the more involved \Cref{thm:edgesinformal}, our decomposition is based on a new process (parameterized by two parameters $b, p$) that crucially differs from the process by~\cite{chang2019distributed} in that nodes of degree at most $p$ may already be removed when they still have neighbors of degree greater than $p$, provided there are not more than $b$ of such neighbors (which also allows us to omit ``rake steps'', addressing nodes of degree $1$, entirely).
As we show, when executed with parameters $p = k = g(n)$ and $b$ chosen suitably from $\Theta(a)$, the decomposition obtained by this process can transformed into a decomposition as described above, via a careful distinction between different kinds of edges (one of which captures all edges between a node $v$ and its neighbors of degree $>p$ at the time at which $v$ was removed).

We note that the two decompositions also differ in how the aforementioned ``parts'' of the decomposition are related: in the decomposition for \Cref{thm:nodesinformal}, the parts are subgraphs that are \emph{node-disjoint} whereas in the decomposition for \Cref{thm:edgesinformal}, they are \emph{edge-disjoint} subgraphs, which, informally speaking reflects the types of problems the respective theorems are applicable to.

We believe that the clean structure of our approach will enable future work to build on it to achieve a similar transformation as captured by \Cref{thm:nodesinformal,thm:edgesinformal} for even wider classes of problems.
In particular, obtaining further decompositions similar to the ones we use seems to be a highly promising approach to this end.

\paragraph{Roadmap.}
In~\Cref{sec:prelims}, we provide some preliminaries, such as graph-theoretic terminology and a formal introduction to the model of computation and the node-edge-checkability formalism.
We will then proceed to prove (the formal version of) \Cref{thm:nodesinformal} in \Cref{sec:provenodes} and (the formal version of) \Cref{thm:edgesinformal} in \Cref{sec:proveedges}, and conclude in \Cref{sec:actualproblems} by proving \Cref{thm:edgedegplusoneintro} and showing how to use (the formal version of) \Cref{thm:edgesinformal} to derive the known upper bound for maximal matching on trees.

\section{Preliminaries}\label{sec:prelims}
In this section, we introduce notation and terminology, in particular graph-theoretic notions, the model of computation, and the classes of problems we will be considering.

Let $G=(V,E)$ be a simple graph. The maximum degree of $G$ is denoted by $\Delta(G)$ (or simply $\Delta$ when $G$ is clear from the context).
For a node $u\in V(G)$ and an edge $e\in E(G)$, denote by $\deg_{G}(u)$ and $\edgedeg_{G}(e)$ (or simply $\deg(u)$ and $\edgedeg(e)$) the number of adjacent nodes of $u$ and the number of adjacent edges of $e$, respectively.
For a subset $P\subseteq V(G)$ and a subset $Q\subseteq E(G)$, denote by $G[P]$ and $G[Q]$ the graphs $(V_P=P,E_P=\{\{u,v\}\in E(G)\mid u,v\in P\})$ and $(V_Q=\{u\in V(G)\mid\exists \  v\ \text{such that }\{u,v\}\in Q\},E_Q=Q)$ respectively.
We may refer to $G[P]$ and $G[Q]$ as the graphs \emph{induced by $P$} and \emph{induced by $Q$}, respectively. 

We use the notion of half-edges which has been used widely in the context of locally checkable problems. For a graph $G=(V,E)$, a pair $(v,e)\in V \times E $ is called a \textit{half-edge} if $e$ is incident on $v$.
We will denote the set of half-edges of $G$ with $H(G)$. A \textit{half-edge labeling} of a graph is a mapping from $H(G)$ to a set of \emph{labels}.

In order to obtain our results in their full generality, we will need to work with objects that are slight extensions of graphs in that they also allows edges with only one or zero endpoints.
For technical reasons, we will define these objects as bipartite graphs, where one side of the bipartition represents the nodes of the object, the other side the edges (with potentially $0$ or $1$ endpoints) of the object and the edges between the two parts the incidence relation between the nodes and edges of the object (i.e., the half-edges of the object).

\begin{definition}[Semi-graph]
    A semi-graph is a bipartite graph $S = (V_S, E_S) = (A\bigsqcup B, C)$ such that for any node $b\in B$, $\deg_{G}(b)\leq 2$.
\end{definition}

Note that also any standard graph $G=(V,E)$ can be understood as a semi-graph $S=(V,E, \{\{v,e\}\mid (v,e)\in H(G)\})$.
We may also use $G$ to refer to this semi-graph when convenient. 

For a semi-graph $S=(A,B,C)$, we will refer to the elements of $A$ as \emph{nodes}, to the elements 
of $B$ as \emph{edges}, and to the elements of $C$ as \emph{half-edges}.
We will denote these as $\Vs(S)$, $\Es(S)$ and $H(S)$, respectively.
For $u\in \Vs(S)$ and $e\in \Es(S)$, we refer to their degree in $S$ by \emph{degree of $u$} and \emph{rank of $e$}, denoted with $\deg{u}$ and $\rank{e}$, respectively.
(Note that the definition of $\deg(\cdot)$ corresponds to the one for graphs given further above in case $S$ is a graph.)
For $h=\{v,e\}\in H(S)$ with 
$v\in\Vs(S)$ and $e\in\Es(S)$, we will say that $h$ is \emph{incident on $u$} (and vice versa), $h$ is \emph{incident on $e$} (and vice versa), and $e$ is \emph{incident on $v$} (and vice versa).

For a semi-graph $S$, define as the \emph{underlying graph of $S$} the graph $G$ where $V(G)=\Vs(S)$, $E(G)=\{\{v_1,v_2\} \mid v_1,v_2\in \Vs(S) \text{ and }\{v_1,e\},\{v_2,e\}\in H(S) \text{ for some } e\in\Es(S)\}$. Denote the degree of the underlying graph of $S$ as the \emph{underlying degree of $S$}. We call a semi-graph $S$ connected if its underlying graph is connected. 

We will abuse terminology and say that \emph{graph $G$ is a subgraph of semi-graph $S$} if $G$ is a subgraph of the underlying graph of $S$. Similarly we will say that \emph{semi-graph $S$ is a subgraph of graph $G$} if $S$ is a subgraph of semi-graph $G$.
In general, a semi-graph should be simply considered as a graph with some additional edges with $0$ or $1$ endpoints, and we may use standard graph-theoretic notions for semi-graphs in the natural way.

Next, we define our model of computation formally.

\begin{definition}[LOCAL model~\cite{linial1992locality,peleg2000}] The LOCAL model of computation is modeled by an undirected graph $G=(V,E)$ in which each of the nodes represents a computational unit with unbounded computational power and memory. Each node knows the number of nodes $n$ and the maximum degree $\Delta$ of the graph. Each of the nodes has a globally unique identifier from $\{1,2\dots, n^c\}$ for some large enough constant $c$. Depending on the problem to solve, there may be other inputs on nodes and edges. A $T$-round LOCAL algorithm runs for $T$ synchronous rounds. In each round, each node first sends (potentially different) messages of arbitrary size to its neighbors, then receives the information sent by its neighbors, and finally may perform some (arbitrarily complex) computation. At the end of T rounds, each node determines its output and terminates. Equivalently, each node learns all the information in its it $T$-hop neighborhood in $T$ rounds (since the message size is unlimited) and decides on an output.
The definition of the LOCAL model extends to semi-graphs in the natural way, where the computational units are the nodes $\Vs(S)$ and messages can only be sent via edges of rank $2$.
\end{definition}

\begin{definition}[Node-edge-checkable problem] A \emph{node-edge-checkable problem} $\Pi$ is a triple $(\Sigma,\mcN_{\Pi},\mcE_{\Pi})$ where $\Sigma$ is a (possibly infinite) set of output labels, $\mcN_{\Pi}=\{\mcN_{\Pi}^{0},\mcN_{\Pi}^{1}, \dots\}$ where $\mcN_{\Pi}^{i}$ is a collection of cardinality-$i$ multi-sets $\{X_1,X_2,\dots, X_i\}$ with $X_1,X_2,\dots, X_i \in \Sigma$, and $\mcE_{\Pi}=\{\mcE_{\Pi}^0,\mcE_{\Pi}^1,\mcE_{\Pi}^2\}$ where $\mcE_{\Pi}^i$ is a collection of cardinality-$i$ multi-sets $\{Y_1,\dots, Y_i\}$ with $Y_1,\dots, Y_i\in\Sigma$.

A valid solution to $\Pi$ on an input semi-graph $S$ is a function $\hout:H(S) \to \Sigma$ such that
\begin{itemize}
    \item for every node $v\in \Vs(S)$, the multiset of labels assigned by $\hout$ to its incident half-edges is in $\mcN_{\Pi}^{\ndeg(v)}$,
    \item for every edge $e\in \Es(S)$, the multiset of labels assigned by $\hout$ to its incident half-edges is in $\mcE_{\Pi}^{\rank(e)}$.
\end{itemize}
\end{definition}

We proceed to define a couple of ``list'' variants of a node-edge-checkable problem $\Pi$ that we will need for stating our theorems.
We will first state them formally and then give some intuition how to understand them.

\begin{definition}[\NodelistP]\label{def:nodelistvariant} For a node-edge-checkable problem $\Pi=(\Sigma,\mcN_{\Pi},\mcE_\Pi)$, define its \emph{node-list variant} $\Pi^*$ as the triple $(\Sigma,\mcL_{\Pi},\mcE_{\Pi})$ where

$\mcL_{\Pi}=\{\mcL_{\Pi}^0, \mcL_{\Pi}^1,\dots\}$ where $\mcL_{\Pi}^i$ is defined as
\begin{itemize}
    \item $\mcL_{\Pi}^i=\{\mcN_{\Pi, \mset}^i \mid \mset\in \mcN_{\Pi}^{j}\text{ for some }j\in\mathbb{Z}\} $ where
    \item for each $j$ and each $\mset=\{\psi_1,\dots,\psi_j\}\in\mcN_{\Pi}^j$, we define $\mcN_{\Pi,\mset}^i=  \{\{\chi_1,\dots ,\chi_i\} \mid\chi_1,\dots,\chi_i\in\Sigma \text{ and }\{\chi_1,\dots ,\chi_i,\mset_1,\dots,\mset_j\} \in \mcN_{\Pi}^{i+j}\}$.
\end{itemize}

An input instance of $\Pi^*$ is a pair $(S,\hin)$ such that
\begin{itemize}
    \item $S$ is a semi-graph, and
    \item $\hin$ is a function $\hin:\Vs(S)\to \bigcup_{j\in\mathbb{Z}_{\geq 0}}\mcL_{\Pi}^j$ such that $\hin(u)\in \mcL_{\Pi}^{\ndeg(u)}$.
\end{itemize} 
A valid solution to $\Pi^*$ on an input instance $(S,\hin)$ is a function $\hout:H(S) \to \Sigma$ such that
\begin{itemize}
    \item for every node $v\in \Vs(S)$, the multiset of labels assigned by $\hout$ to its incident half-edges is in $\hin(u)$, and
    \item for every edge $e\in \Es(S)$, the multiset of labels assigned by $\hout$ to its incident half-edges is in $\mcE_{\Pi}^{\rank(e)}$.
    
\end{itemize}  
\end{definition}

\begin{definition}[\EdgelistP]\label{def:edgelistvariant} For a a node-edge checkable problem $\Pi=(\Sigma,\mcN_{\Pi},\mcE_\Pi)$, define its \emph{edge-list variant} $\Pi^{\times}$ as the triple $(\Sigma,\mcN_{\Pi},\mcL_{\Pi})$ where

$\mcL_{\Pi}=\{\mcL_{\Pi}^0,\mcL_{\Pi}^1,\mcL_{\Pi}^2\}$ where $\mcL_{\Pi}^i$ is defined as
\begin{itemize}
    \item $\mcL_{\Pi}^i=\{\mcE_{\Pi,\mset}^i \mid \mset\in \mcE_{\Pi}^{j}\text{ for some }0\leq j\leq 2-i\} $ where
    \item for each $j$ and each $\mset=\{\mset_1,\dots,\mset_j\}\in\mcE_{\Pi}^j$, we define $\mcE_{\Pi,\mset}^i=  \{\{\chi_1, \dots ,\chi_{i}\} \mid \chi_1,\dots ,\chi_{i}\in\Sigma \text{ and }\{\chi_1,\dots ,\chi_i,\mset_1,\dots,\mset_j \in \mcE_{\Pi}^{i+j}\}.$
\end{itemize}

An input instance of $\Pi^{\times}$ is a pair $(S,\hin)$ such that
\begin{itemize}
    \item $S$ is a semi-graph, and
    \item $\hin$ is a function $\hin:\Es(S)\to \bigcup_{j\in \{0,1,2\}}\mcL_{\Pi}^j$ such that $\hin(e)\in \mcL_{\Pi}^{\rank(e)}$.
\end{itemize} 
A valid solution to $\Pi^{\times}$ on an input instance $(S,\hin)$ is a function $\hout:H(S) \to \Sigma$ such that
\begin{itemize}
    \item for every node $v\in \Vs(S)$, the multiset of labels assigned by $\hout$ to its incident half-edges is in $\mcN_{\Pi}^{\ndeg(v)}$, and
    \item for every edge $e\in E_{semi}(S)$, the multiset of labels assigned by $\hout$ to its incident half-edges is in $\hin(e)$.

\end{itemize}  
\end{definition}

Intuitively, the problems defined in \Cref{def:nodelistvariant,def:edgelistvariant} should be understood as problems arising on a subsemi-graph of the input semi-graph $S$ when the actual problem $\Pi$ that is supposed to be solved has already been solved on the part of $S$ not belonging to the subsemi-graph.
Imagine an output for $\Pi$ has been fixed on some half-edges such that
\begin{itemize}
    \item for each edge (of any rank), either all incident half-edges received an output label or none of them did, and
    \item for each node and each edge of $S$ for which the output has been fixed on all incident half-edges, the output is correct (i.e., is contained in the respective $\mcN_{\Pi}^i$ or $\mcE_{\Pi}^i$).
\end{itemize}
Consider a node $v \in \Vs(S)$.
Some of the half-edges incident to $v$ might have received their output while others might not have.
If we wanted to characterize which possibilities are left for the labels on those incident half-edges that have not received their output, we might do so by assigning an input label to $v$ that lists all of these possibilities.
This input label would essentially be a new collection of allowed label configurations that replaces the node constraint $\mcN_{\Pi}^i$ corresponding to $v$'s degree.

Now Problem $\Pi^*$ can be seen as the problem that captures the task of solving all instances where each node is given such a new collection of label configurations that might come from the outlined scenario that the instance is actually part of a larger instance that has been partially solved.
From this perspective, $i$ in~\Cref{def:nodelistvariant} represents the ``remaining unsolved'' degree of a considered node, $j$ the ``solved'' degree, $i + j$ the original degree (in the larger instance), $\psi$ a possible output on the ``solved'' incident half-edges, $\mcL_{\Pi}^i$ the set of all possible replacement collections of configurations, each of which comes from a possible way how the solved incident half-edges are fixed (specified by $\psi$), $\hin$ the function assigning to each node its replacement collection, and so on. 

The intuition for~\Cref{def:edgelistvariant} is similar, with the essential difference being that the scenario to be imagined comes with the requirement that for each \emph{node} (instead of for each \emph{edge}), either all incident half-edges received an output label or none of them did.

For examples of node-list and edge-list variants of concrete problems, we refer to~\Cref{sec:actualproblems}.

\section{Proving \Cref{thm:nodesinformal}}\label{sec:provenodes}
A key technique used in this section is the rake-and-compress technique on trees introduced in \cite{miller1985parallel}. There are a number of variants of this method (see, e.g., \cite{chang2019distributed,lievonen2024distributed}), one of which~\cite{chang2019distributed} we will rely on heavily in this section and discuss below. The process takes as input a graph which is a tree and a parameter $k\geq2$ and outputs a partition of the nodes into layers. The process proceeds by iterative application of rake and compress operations which are defined as follows:
\begin{itemize}
    \item \textit{Rake operation}: Rake($G$) marks a node $u\in V(G)$ if the degree of $u$ is at most 1.
    \item \textit{Compress operation}: Compress($G,k$) marks a node $u\in V(G)$ if the degrees of $u$ and all of its neighbors are at most $k$.
\end{itemize}
  We say a node is \textit{raked} if it is marked by a rake operation and a node is \textit{compressed} if it is marked by a compress operation. 
  We may refer to them as raked nodes and compressed nodes respectively. We now formally define the rake-and-compress algorithm on trees from~\cite{chang2019distributed}.
\begin{algorithm}[H]
        \caption{Rake-and-Compress Algorithm}\label{rake-and-compress}
        \begin{algorithmic}
           
            \State Let $G=(V,E)$ be a tree. Set all nodes of $G$ to be unmarked initially.  
            \State Set $V_0=V(G)$. For $i\in\{1,2\dots, \ceil{\log_{k}n+1}$, do the following:
                    \begin{enumerate}
                    \item Set $C_i=\emptyset$, $R_i=\emptyset$. 
                    \item Perform Compress($G[V_{i-1}],k$). Add the nodes marked in this compress step to $C_i$. 
                    \item Perform  Rake($G[V_{i-1}\setminus C_i]$). Add the nodes marked in this rake step to $R_i$. 
                    \item Set $V_i$ to be the set of unmarked nodes in $V_{i-1}$.
                    \end{enumerate} 
                    
        \end{algorithmic}
    \end{algorithm}

Theorem 9 in \cite{chang2019distributed} guarantees that \cref{rake-and-compress} is sufficient to mark all the nodes.
\begin{lemma}[\cite{chang2019distributed}]
All nodes in the input tree are either raked or compressed by \cref{rake-and-compress}.
\end{lemma}
We will refer to the sets $C_i$ and $R_i$ as \textit{compress layers} and  \textit{rake layers} respectively. We will use simply the term \textit{layer} to refer to any of the compress and rake layers. We fix a total ordering on the layers in the following way: for two distinct layers $L_1$ and $L_2$, we say that \emph{$L_1$ is higher than $L_2$} if nodes in $L_1$ were marked after nodes in $L_2$. We may equivalently say \emph{$L_2$ is lower than $L_1$}. Similarly, we fix a total ordering on the nodes in the following way: \emph{node $u$ is lower than node $v$}, or equivalently \emph{node $v$ is higher than node $u$} if node $u$ is in a lower layer than node $v$. If both nodes are in the same layer, the one with higher ID is considered to be the higher node. If $u$ and $v$ are neighbors with $u$ higher than $v$, call $u$ a higher neighbor of $v$, and $v$ a lower neighbor of $u$. For an edge $e=\{u,v\}$, call $v$ the higher endpoint of $e$ and $u$ the lower endpoint of $e$ if $u$ is lower than $v$.

We now claim that the subtree of the input tree induced by the edges that have their lower endpoint in a compress layer has maximum degree bounded by $k$.
    \begin{lemma} \label{lemma-degree-bound}
    The maximum degree of the graph induced by the edges that have their lower endpoint in one of the layers $C_1,C_2,\dots, C_{\ceil{\log_k n+1}}$ is at most $k$.
   \end{lemma}
   \begin{proof}
    Let $E_C$ denote the set of those edges that have their lower endpoint in one of the layers $C_1,C_2\dots, C_{\ceil{\log_k n+1}}$. Assume for a contradiction that there is a node $u\in V(G[E_C])$ with $\deg_{G[E_C]}(u)>k$. The design of \cref{rake-and-compress} implies that $u$ has at most $k$ higher neighbors. Hence, there is at least one lower neighbor of $u$. Since all edges of $G[E_C]$ have their lower nodes in some compress layer, all lower neighbors of $u$ are in compress layers. Let $v\in C_j$, for some $j$, be the lowest among the neighbors of $u$. This implies that when $v$ was marked, $u$ still had more than $k$ unmarked neighbors, which implies that $v$ would not have been marked in the respective compress operation, yielding a contradiction. 
\end{proof}
Next we bound the diameter of the connected components in the graph induced by the raked nodes.
\begin{lemma}\label{diameter-non-special-part} 
        The diameter of a connected component in the graph induced by the raked nodes is at most $4(\log_k n+1)+2$.
    \end{lemma}
    
    \begin{proof}
         Let $R$ be the set of raked nodes, and let $u$ be a node in a connected component $\mathcal{P}$ of $G[R]$. Let $r$ be the highest node in $\mathcal{P}$. There must exist such an $r$ since the ordering induced on the nodes is a total ordering. We will show that $u$ is at distance at most $2(\log_k n+1)+1$ from $r$ which would prove the claim. If $r= u$, we are done. Assume otherwise. Consider the unique path $p=(u=v_1,v_2,\dots, v_{\ell}=r)$ of length $\ell-1$ from $u$ to $r$. Let $j\in\{1,2,\dots, \ell-2\}$. Suppose $v_{j}$ is higher than $v_{j+1}$. Since $v_{j+1}$ has at most one neighbor higher than $v_{j}$ (since $v_j$ is in a rake layer), $v_{j+2}$ must be lower than $v_{j+1}$, and therefore also lower than $v_{j}$. By iterating this argument with increasing $j$, we obtain that $v_{\ell}=r$ is lower than $v_j$, a contradiction. Thus, $v_j$ is lower than $v_{j+1}$, which in turn is lower than $v_{j+2}$.
         
         Moreover, if $v_{j}$ and $v_{j+1}$ are in the same layer, $v_{j+1}$ would not have a higher neighbor (since $v_{j+1}$ is in a rake layer and $v_j$ and $v_{j+1}$ being raked at the same time means they do not have neighbors in higher layers), implying that $v_{j+1}$ is $r$. This contradicts the fact that $j\leq \ell-2$. Hence, the nodes $v_1,v_2,\dots,v_{\ell-1}$ are in different layers. Since there are at most $2(\log_k n+1)$ layers, it follows that $\ell-1\leq 2(\log_k n +1)+1$ and that $u$ is at distance at most $2(\log_k n +1)+1$ from $r$.
    \end{proof}
    
We are now ready to state and prove (the formal version of) \Cref{thm:nodesinformal}.

\begin{theorem}\label{thm:nodesveryformal}
    Let $\Pi$ be a node-edge-checkable problem and $f$ a monotonically non-decreasing, non-zero function.
    Assume that
    \begin{itemize}
        \item $\Pi$ admits an algorithm $\mathcal{A}$ of complexity $O(f(\Delta)+\log^*n)$ rounds on semi-graphs where $\Delta$ is the degree of the underlying graph, and 
        \item $\Pi^{\times}$ admits a valid solution on any valid input instance.
    \end{itemize}
    Then $\Pi$ can be solved in $O(f(g(n)) + \log^* n)$ rounds on trees, where $g$ is the function satisfying $g(n)^{f(g(n))} = n$.
\end{theorem}
\begin{proof}
    Let $T$ be a tree with $n$ nodes and set $k := g(n)$.
    In the following we describe our algorithm for~\Cref{thm:nodesveryformal}.
    
    Apply \cref{rake-and-compress} on tree $T$ (with parameter $k$). Let $C_1,C_2,\dots, C_{\log_{k}n+1}$ be the obtained compress layers and $R_1,R_2,\dots, R_{\log_{k}n+1}$ the obtained rake layers. Let $C$ be the set of compressed nodes and $R$ the set of raked nodes.
    Let $T_C$ be the semi-graph with node set $\Vs(T_C) = C$, edge set $\Es(T_C) = \{ e \in E(T) \mid \exists v \in C \text{ such that } v \in e \}$ and half-edge set $H(S) = \{ \{ v, e \} \mid v \in C, e \in \Es(T_C), v \in e \}$.
    Let $T_R$ be the semi-graph obtained analogously (by replacing $R$ with $C$ in the definition of $T_C$).
    Note that some edges that have rank $2$ in $T$ might have lower rank in $T_C$ or $T_R$ (similarly to the situation for nodes and their degrees).
    
    We now run \cref{node-solving algorithm}, given below.
    \begin{algorithm}
        \caption{Edge-list solver}\label{node-solving algorithm}
        \begin{algorithmic}[1]
            \Statex \textbf{Input}: Semi-graphs $T_C,T_R$
            \State Run Algorithm $\mathcal{A}$ on the semi-graph $T_C$. 
            \State Consider the semi-graph $T_R$.
            \begin{itemize}
                \item For each $r \in \{ 0, 1, 2 \}$ and each edge $e\in \Es(T_R)$ with rank $r$, let $\chi(e)$ be the multiset of output labels that have already been assigned to half-edges incident on $e$ in $T$.
                 Set $\hin(e):=\mcE_{\Pi,\chi(e)}^{r}$.
                \item Consider a connected component $\mathcal{C}$ in $T_R$ and its underlying graph $\mathcal{C}'$. Now let the highest node in $\mathcal{C}'$ collect the entire connected component, compute a correct solution of $\Pi^{\times}$ on $(\mathcal{C},{\hin}_{\mid \mathcal{C}})$ (where ${\hin}_{\mid \mathcal{C}}$ denotes function $\hin$ restricted to $\mathcal C$), and inform all other nodes in $\mathcal{C}'$ about the solution (upon which each node outputs its local part of the solution). Let the algorithm that solves $\Pi^{\times}$  on $T_R$ be called $\mathcal{B}$.
            \end{itemize}
        \end{algorithmic}
    \end{algorithm}
    This concludes the description of our algorithm.
    In the following we analyze its runtime and show that the obtained output is correct.
    We start with the runtime.
    
     Observe that the underlying graph of semi-graph $T_C$ is the subgraph of graph $T$ induced by the compressed nodes in $T$. Hence, by \cref{lemma-degree-bound}, the degree of the underlying graph of $T_C$ is at most $k$.
     It follows that Line 1 of \cref{node-solving algorithm} fixes the output labels of the half-edges in $T_C$ in $O(f(g(n)+\log^*n)$ rounds.  Note that the underlying graph of semi-graph $T_R$ is the subgraph of graph $T$ induced by the raked nodes in $T$. Hence, by \cref{diameter-non-special-part}, Line 2 of \cref{node-solving algorithm} runs in $O(\log_k n) = O(f(g(n)))$ rounds.
     As each of the $O(\log_k n)$ iterations in the execution of \cref{rake-and-compress} can be performed in a constant number of rounds, we therefore obtain an overall runtime of $O(f(g(n)+\log^*n)$ rounds.

     Consider a node $u$ and let $\chi$ be the multiset of output labels assigned by our algorithm to the half-edges incident on $u$. Refer to $\chi$ as the \emph{node configuration of $u$}. We say $\chi$ is a \emph{valid node configuration of $u$} for problem $\Pi$ if $\chi\in\mcN_{\Pi}^{\ndeg(u)}$. Similarly, consider an edge $e$ and let $\mset$ be the multiset of output labels assigned by our algorithm to the half-edges incident on $e$. Refer to $\mset$ as the \emph{edge configuration of $u$}. We say $\mset$ is a \emph{valid edge configuration of $e$} if $\mset\in\mcE_{\Pi}^{\rank(e)}$. Hence, a solution to $\Pi$ is correct if all the nodes and edges have valid configurations. We now proceed to show the correctness of \cref{node-solving algorithm} by going through the following exhaustive cases.
     \begin{enumerate}[label=(\roman*)]
        \item Assume there is a node $u$ of semi-graph $T$ on which the node configuration is not valid for $\Pi$. Since \cref{node-solving algorithm} solves $\Pi$ on $T_C$ and $\Pi^{\times}$ on $T_R$, all the half-edges incident on $u$ are labeled entirely by either algorithm $\mathcal{A}$ or by Algorithm $\mathcal{B}$. Algorithm $\mathcal{A}$ and Algorithm $\mathcal{B}$ correctly solve problem $\Pi$ and $\Pi^{\times}$, respectively (due to the assumptions made in \Cref{thm:nodesveryformal}), and hence output for $u$ a configuration that is in $\mcN_{\Pi}^{\ndeg(u)}$. This yields a contradiction.
        \item Assume there is an edge $e$ of semi-graph $T$ on which the edge configuration is not valid for $\Pi$. Suppose Algorithm $\mathcal{A}$ labels $q$ half-edges incident on $e$ and let the set of these labels be $\chi$. Since algorithm $\mathcal{A}$ solves $\Pi$, $\chi\in \mcE_{\Pi}^q$. If $q=2$, then the output of Algorithm $\mathcal{A}$ (and therefore also of the overall algorithm) for $e$ is $\chi\in \mcE_{\Pi}^2$ which implies that the edge configuration on $e$ is valid. Hence, $q<2$ and Algorithm $\mathcal{B}$ labels the $2-q$ yet unlabeled half-edges incident on $e$ which forms a multiset $\mset$. By definition, $\mset\in \hin(e)=\mcE_{\Pi,\chi}^q$. It follows that $\mset\cup\chi\in\mcE_{\Pi}^{q+2-q}=\mcE_{\Pi}^2$ and that the edge configuration on $e$ is valid. This yields a contradiction.

    \end{enumerate}
\end{proof}

\section{Proving \Cref{thm:edgesinformal}}\label{sec:proveedges}
Similar to our approach in \Cref{sec:provenodes}, we start by developing a process \emph{Decomposition} to partition the nodes of the input graph into layers. The process takes as input a graph of arboricity $a$ and two integer parameters $b$ and $k$ such that $a<b$ and $5a\leq k$. It outputs a partition of the nodes into layers $C_1,C_2\dots, C_{\kappa\log_{\frac{k}{a}}n}$ such that $u\in C_i$ has at most $k$ neighbors in $\bigcup_{j\geq i}C_j$ and at most $b$ neighbors in $\bigcup_{j\geq i}C_j$ with degree at least $k+1$ in $G[\bigcup_{j\geq i}C_j]$. The core component of the \emph{Decomposition} process is the repeated iteration of the following \emph{Compress} operation:

\emph{Compress($G,b,k$)}: Mark a node $u\in V(G)$ if the degree of $u$ is at most $k$ and at most $b$ neighbors of $u$ have degree greater than $k$.

\begin{algorithm}
    \caption{Decomposition}\label{Decomposition}
    \begin{algorithmic}[1]
        \Statex \textbf{Input}: A graph $G$ with an upper bound $a$ for the arboricity of $G$, and integers $b$ and $k$ such that $a<b$ and $5a\leq k$.
        \State All nodes are unmarked initially.
        \State Set $V_0\coloneqq V(G)$.
        \State For $i=1,2\dots,\lceil{\kappa}\log_{\frac{k}{a}}n\rceil + 1$:
                \begin{itemize}
                    \item Set $C_i\coloneqq\emptyset$.
                    \item Run Compress($G[V_{i - 1}],b,k$).
                    \item Set $C_i$ to be the set of nodes marked in this iteration. Set $V_i\coloneqq V_{i-1}\setminus C_i$.
                \end{itemize}
    \end{algorithmic}
\end{algorithm}
We now show that \cref{Decomposition} marks all nodes in $G$ reasonably fast for the parameter choice $b = 2a$.

\begin{lemma}\label{lem:newdecompruntime}
    Let $G$ be a graph on n nodes with arboricity at most $a$. Given integers $b=2a$ and $k$ such that $5a \leq k$, \cref{Decomposition} marks all nodes of $G$ in $\lceil{\kappa}\log_{\frac{k}{a}}n\rceil + 1$ rounds.
   
\end{lemma}
\begin{proof}
    Consider sets $V_i$ and $V_{i+1}$ for any $i\in\{0,1,\dots,\lceil{\kappa}\log_{\frac{k}{a}}n\rceil\}$. Let $n_i=|V_i|$ and $n_{i+1}=|V_{i+1}|$. Let $u\in V_i$ be called a \emph{large node} if $\deg_{G[V_{i}]}(u)>k$ and a \emph{small node} otherwise. Consider the following exhaustive cases.
    \begin{enumerate}[label=(\roman*)]
        \item Case 1: $V_{i+1}$ contains at least $\frac{n_{i+1}}{2}$ large nodes.

              As the arboricity of $G[V_i]$ is at most $a$, the number of edges in $G[V_{i}]$ is at most $a\cdot(n_{i}-1)$. Let the number of edges in $G[V_i]$ incident on at least one large node be $d$. Then,
             \[
                  a\cdot(n_{i}-1) \geq d
                  \geq \frac{1}{2}\cdot \frac{n_{i+1}}{2}\cdot k,
            \]
            which implies
            \[
                   n_{i+1} \leq \frac{4a\cdot(n_{i}-1)}{k}.
            \]

          \item Case 2: $V_{i+1}$ contains at least $\frac{n_{i+1}}{2}$ small nodes.
                 
                 Let $p$ be the number of small nodes that are contained in $V_{i+1}$.   
                 Consider the bipartite subgraph $H$ of $V_i$ induced by the edges between small nodes that are contained in $V_{i + 1}$ and large nodes. If a small node is adjacent to fewer than $b$ large nodes in $V_{i}$, it is added to $C_{i+1}$ and hence not contained in $V_{i+1}$. Therefore, a small node that is contained in $V_{i+1}$ has at least $b$ large nodes as neighbors in $V_{i}$. Hence, there are at least $b\cdot p$ edges in $H$. Since $H$ has arboricity at most $a$, there must be at least $\frac{b\cdot p}{a}$ nodes in $H$. Since there are precisely $p$ small nodes in $H$, there are at least $p\cdot(\frac{b}{a}-1)$ large nodes in $H$. Let the number of edges in $G[V_i]$ incident on at least one large node be $d$. Then
                 \begin{align*}
                     & &a\cdot(n_{i}-1)  &\geq d& \\
                     &&&\geq \frac{p\cdot k}{2}\cdot \left(\frac{b}{a}-1\right)\\
                     &&&\geq \frac{k}{2}\cdot \frac{n_{i+1}}{2}\left(\frac{b}{a}-1\right),    
                 \end{align*}
                 which implies
                 \[n_{i+1} \leq \frac{4a\cdot(n_{i}-1)}{k}.\]
    \end{enumerate}
    Hence, in either case, the number of unmarked nodes reduces by a factor of at least $\frac{k}{4a}$ after each application of the compress operation. This implies that the algorithm marks all nodes in $G$ after $\log_{\frac{k}{4a}}n + 1 \leq \lceil{\kappa}\log_{\frac{k}{a}}n\rceil + 1$ iterations (where the inequality is due to $\log_{5/4} 5 \leq 10$).
    Since each iteration can be performed in one round, we obtain the lemma statement.
\end{proof}

Before proceeding to the statement and proof of the formal version of \Cref{thm:edgesinformal}, we collect some properties of the decomposition returned by \Cref{Decomposition}.
We start by defining the necessary terminology.

We will refer to the sets $C_i$ as \emph{layers}. We fix a total ordering on the 
layers in the following way: for two distinct layers $L_1$ and $L_2$, 
we say that \emph{$L_1$ is higher than $L_2$} if the nodes in $L_1$ were marked after the nodes 
in $L_2$. We may equivalently say that \emph{$L_2$ is lower than $L_1$}. 
Similarly, we fix a total ordering on the nodes in the following way: for 
two distinct nodes $u$ and $v$, \emph{$u$ is lower than $v$} (or 
equivalently \emph{$v$ is higher than $u$}) if $u$ is in a lower layer than 
$v$. If both $u$ and $v$ are in the same layer, the one with higher ID is considered to be the higher node.
For an edge $e=\{u,v\}$, if $u$ is lower than $v$, 
call $v$ the \emph{higher endpoint of $e$} and $u$ the \emph{lower endpoint 
of $e$} .

Consider a graph $G$ of arboricity at most $a$. Let $k$ be a parameter greater than $5a$ and let $b = 2a$. Apply \cref{Decomposition} on $G$ and let $C_1,C_2\dots, C_{\lceil\kappa\log_{\frac{k}{a}}n\rceil + 1}$ be the partition of the nodes obtained from the algorithm. Consider an edge $e=\{u,v\}$ with $u\in C_i$, $v\in C_j$, and $i<j$. Call $e$ \emph{atypical for $u$} if $\deg_{G[V_{i-1}]}(v)>k$. Note that the compress operation ensures that there are at most $2a$ edges that are atypical for $u$. Call an edge \emph{atypical} if it is atypical for its lower endpoint and call the rest of the edges \emph{typical}. Let $E_1$ be the set of atypical edges and $E_2$ the set of typical edges.
    
     \begin{lemma}\label{degree-bound on arb graphs}
         The graph induced by the typical edges has maximum degree at most $k$.
     \end{lemma}
     \begin{proof}
     
     Consider the graph $G[E_2]$ induced by $E_2$. The lemma claims that  $\Delta(G[E_2])\leq k$.
     Suppose for a contradiction that $\deg_{G[E_2]}(u)\geq k+1$ for some $u$, and let $i$ be the index such that $u\in C_i$. If $u$ has no neighbors in $G[E_2]$ lower than $u$, then $\deg_{G[V_{i-1}]}(u) \geq k+1$, which would imply $u\notin C_i$, a contradiction. Therefore, $u$ has a 
     lower neighbor in $G[E_2]$. Let $v\in C_j$ be the lowest neighbor of 
     $u$ such that $e=\{u,v\}\in E_2$. This implies that $e$ is 
     an atypical edge for $v$, which implies $e\notin E_2$, a 
     contradiction. Hence $\Delta(G[E_2])\leq k$. 
     \end{proof}   
     
    Any node $u \in G[E_1]$ has at most $2a$ higher neighbors. Create $2a$ edge-disjoint graphs $G[F_1],G[F_2],\dots, G[F_{2a}]$ as follows.
     \begin{enumerate}
         \item With colors from the palette $\{1,2,\dots, 2a\}$, each node $u\in G[E_1]$ colors its incident edges to higher neighbors differently. This produces a coloring $c$ (which may not be proper) on the edges in $G[E_1]$.
         \item For $i\in\{1,2,\dots, 2a\}$, define $F_i=\{e\in E_1 \mid c(e)=i\}$.
     \end{enumerate}

     Note that each node in $G[F_i]$ has at most one higher neighbor in $G[F_i]$ and none in the same layer. Hence each of the $G[F_i]$ is a forest and a $3$-coloring $c_i$ of its vertices can be computed in $O(\log^*n)$ rounds~\cite{goldberg1987parallel}. We use this to create a partition of $G[F_i]$ as follows.
     \begin{itemize}
         \item For $i\in\{1,2,\dots 2a\}$ and $j\in\{1,2,3\}$, define $F_{i,j}=\{e\in F_i \mid c_i(\text{higher endpoint of } e)=j\}$.
     \end{itemize}
     
Note that each connected component of any $G[F_{i,j}]$ is a star graph with the highest node in the connected component being the center of the connected component.

We are now ready to state and prove (the formal version of) \Cref{thm:edgesinformal}, which is parameterized by some integer $\rho$.

\begin{theorem}\label{thm:edgesveryformal}
    Let $\Pi$ be a node-edge-checkable problem and $f$ a monotonically non-decreasing, non-zero function. Assume that
    \begin{itemize}
        \item $\Pi$ admits an algorithm $\mathcal{A}$ of complexity $O(f(\Delta)+\log^*n)$ rounds on semi-graphs where $\Delta$ is the underlying degree of the graph, and 
        \item $\Pi^{*}$ admits a valid solution on any valid input instance.
    \end{itemize}
    Then, for any positive integer $\rho$, problem $\Pi$ can be solved in $O\big(a+\frac{\rho \cdot f(g(n)^{\rho})}{\rho-\log_{g(n)}a}+\log^{*}n\big)$ rounds on graphs of arboricity at most $a\leq \frac{g(n)^\rho}{5}$, where $g$ is the function that satisfies $g(n)^{f(g(n))}=n$.
\end{theorem}
\begin{proof}
    Let $G$ be a graph (i.e., also a semi-graph) with $n$ nodes and set $k := g(n)^{\rho}$ and $b := 2a$.
    In the following, we describe our algorithm for~\Cref{thm:edgesveryformal}.

    We start by applying~\Cref{Decomposition} on $G$.
    Let $E_1$ and $E_2$ be as defined above.
    We continue by creating, for any $i\in\{1,2,\dots2a\}$ and $j\in\{1,2,3\}$, the subgraphs $G[F_{i,j}]$ as described above. 
    Then we run \cref{edge-solving algorithm} on semi-graph $G$, which is as given in the following.
    \begin{algorithm}
        \caption{Node-list solver}\label{edge-solving algorithm}
        \begin{algorithmic}[1]
            \Statex \textbf{Input}: Semi-graphs $G[E_2],G[F_{1,1}],\dots,G[F_{2a,3}]$.
            \State Run algorithm $\mathcal{A}$ on the semi-graph $G[E_2]$. 
            \State For $i = 1,2,\dots, 2a$ and $j\in\{1,2,3\}$, do the following:
            \begin{itemize}
                \item For each $d\in\{1,2,\dots,\Delta\}$ and each node $u\in \Vs(G[F_{i,j}])$ with degree $d$, let $\chi(u)$ be the multiset of output labels that have already been assigned to half-edges incident on $u$ in $G$. Set $\hin^{i,j}(u) := \mcN_{\Pi,\chi(u)}^{d}$.
                \item Consider a connected component $\mathcal{C}$ in semi-graph $G[F_{i,j}]$. Let the highest node in $\mathcal{C}$ collect the entire connected component and compute a valid solution to $\Pi^{*}$ on $(\mathcal{C},{\hin^{i,j}}_{\mid \mathcal{C}})$ (where ${\hin^{i,j}}_{\mid\mathcal{C}}$ is the restriction of $\hin^{i,j}$ to the component $\mathcal{C}$) and send the solution to all other nodes in the component. Each node subsequently outputs its local part of the computed solution. Let the algorithm that solves $\Pi^*$ on $(G[F_{i,j}],{\hin^{i,j}})$ be called $\mathcal{A}_{i,j}$.
            \end{itemize}
        \end{algorithmic}
    \end{algorithm}
    
    This concludes the description of our algorithm.
    In the following, we analyze its runtime and show that the obtained output is correct.
    We start with the runtime.
    
    Line 1 in \cref{edge-solving algorithm} runs in $O(f(k)+\log^{*}n)$ rounds since the underlying graph of semi-graph $G[E_2]$ (which does not contain edges of rank $\neq 2$) is graph $G[E_2]$ whose maximum degree is at most $k$. Similarly, the underlying graph of semi-graph $G[F_{i,j}]$ is graph $G[F_{i,j}]$ and hence each loop of Line 2 in \cref{edge-solving algorithm} runs in a constant number of rounds as any connected component of any $G[F_{i,j}]$ is a star graph.
    Moreover, the time to create the subgraphs $G[F_{i,j}]$ is $O(\log^* n)$ rounds and the execution of \Cref{Decomposition} terminates after $O(\log_{\frac{k}{a}}n)$ rounds, due to \Cref{lem:newdecompruntime}.
    Hence, the overall runtime is 
    \begin{align*}
        O(f(k)+a+\log^{*}n+\log_{\frac{k}{a}}n) &=O\left(a+f(k)+\frac{\log_{k}n}{\log_{k}\frac{k}{a}}+\log^*n\right)\\
        &\subseteq O\left(a+f(k)+\frac{f(k)}{\log_{k}\frac{k}{a}}+\log^*n\right)\\
        &=O\left(a+f(k)+\frac{f(k)}{1 - \log_{k} a}+\log^*n\right)\\
                             &= O\left(a+\frac{\rho \cdot f(g(n)^{\rho})}{\rho-\log_{g(n)}a}+\log^*n\right).
    \end{align*}
    
    The inclusion in the second line follows since $k=g(n)^\rho$ and hence $\log_kn \leq \log_{g(n)}n = f(g(n))\leq f(k)$.
    
    Consider a node $u$ and let $\chi$ be the multiset of output labels assigned by our algorithm to the half-edges incident on $u$. Refer to $\chi$ as the \emph{node configuration of $u$}. We say $\chi$ is a \emph{valid node configuration of $u$} for problem $\Pi$ if $\chi\in\mcN_{\Pi}^{\ndeg(u)}$. Similarly, consider an edge $e$ and let $\mset$ be the multiset of output labels assigned by our algorithm to the half-edges incident on $e$. Refer to $\mset$ as the \emph{edge configuration of $u$}. We say $\mset$ is a \emph{valid edge configuration of $e$} if $\mset\in\mcE_{\Pi}^{\rank(e)}$. Hence, a solution to $\Pi$ is correct if all the nodes and edges have valid configurations. We now proceed to show the correctness of \cref{edge-solving algorithm} by going through the following exhaustive cases.
    \begin{enumerate}[label=(\roman*)]
        \item Assume there is an edge $e$ of semi-graph $G$ on which the edge configuration is not valid for $\Pi$. Since \cref{edge-solving algorithm} solves $\Pi$ on $G[E_2]$ and $\Pi^*$ on all $G[F_{i,j}]$, all the half-edges incident on $e$ are labeled entirely by either Algorithm $\mathcal{A}$ or by Algorithm $\mathcal{A}_{i,j}$ for some $i$ and $j$.
        Algorithm $\mathcal{A}$ and Algorithms $\mathcal{A}_{i,j}$ correctly solve problems $\Pi$ and $\Pi^{*}$, respectively (due to the assumptions in \cref{thm:edgesveryformal}), and thus output for $e$ a configuration in $\mcE_{\Pi}^{2}$.
        This yields a contradiction.
        
        \item Assume there is a node $u$ of semi-graph $G$ on which the node configuration is not valid for $\Pi$. 

        Suppose Algorithm $\mathcal{A}$ has assigned labels to $r$ half-edges incident on $u$. Let the multiset of these labels be $\chi$. Since Algorithm $\mathcal{A}$ solves $\Pi$, $\chi\in\mcN_{\Pi}^{r}$. If $r=\ndeg(v)$, then the output of Algorithm $\mathcal{A}$ (and therefore also of the overall algorithm) for $u$ is   
        $\chi\in\mcN_{\Pi}^{\ndeg(v)}$, which implies $u$ has a valid node configuration. Hence $r<\ndeg(v)$.

        Let ${(i',j')}$ be the pair of indices corresponding to the earliest iteration in which $\mathcal{A}_{i',j'}$ assigns labels to $p' > 0$ half-edges incident on $u$.
        Let the multiset of those $p'$ labels be $\mset'$. Suppose that $r'$ half-edges incident on $u$ are already labeled beforehand. Let the multiset of those $r'$ labels be $\chi'$. We know that $\chi'\in\mcN_{\Pi}^{r'}$ due to the correctness of Algorithm $\mathcal A$. By definition, $\mset'\in\hin^{i',j'}(u)=\mcN_{\Pi,\chi'}^{p'}$) which implies that $\mset'\cup\chi'\in\mcN_{\Pi}^{r'+p'}$. If $r'+p'=\ndeg(u)$, the node configuration is valid. Hence, $r'+p'<\ndeg(u)$.
        Now let ${(i'',j'')}$ be the pair of indices corresponding to the next iteration in which $\mathcal{A}_{i'',j''}$ assigns labels to $p'' > 0$ half-edges incident on $u$. Algorithm $\mathcal{A}_{j''}$ can assume that the multiset of the labels of the $r'+p$ half-edges already labeled is in $\mcN_{\Pi}^{r'+p'}$ due to the correctness of Algorithms $\mathcal A$ and $A_{i',j'}$. By iterating these steps, we will obtain inductively that the final output assigned to $u$ when all algorithms $A_{i,j}$ have terminated is a valid node configuration of $u$. 
        This yields a contradiction.
 \end{enumerate}
\end{proof}    

\section{Implications for Concrete Problems}\label{sec:actualproblems}
In this section, we will finally apply~\Cref{thm:edgesveryformal} to obtain our new upper bound for $(\text{edge-degree} + 1)$-edge coloring.
We also show how our techniques can be used to obtain in a generic manner the tight $O(\log n/\log \log n)$-round upper bound for maximal matching known from~\cite{barenboim2013distributed}.

For both problems, our first step is to phrase the problems in the node-edge-checkability formalism for semi-graphs defined in~\Cref{sec:prelims}.
In other words, for either of the two problems we present a problem given in the formalism and then show that it is equivalent to $(\text{edge-degree} + 1)$-edge coloring, resp.\ maximal matching, in the sense that a solution to the problem given in the formalism on a semi-graph can be transformed in at most $1$ round into a solution to the edge coloring or matching problem on the underlying graph and vice versa.
Moreover, in order to show that the defined problems satisfy the conditions of ~\Cref{thm:edgesveryformal} we show that they can be solved in time $O(f(\Delta) + \log^* n)$ for suitable functions $f$, and that their node-list variants admit a valid solution on any valid instance.
We cover $(\text{edge-degree} + 1)$-edge coloring in~\Cref{sec:applicationedge} and maximal matching in \Cref{sec:applicationmatching}.

\subsection{$(\text{Edge-degree} + 1)$-Edge Coloring}\label{sec:applicationedge}
We define a node-edge-checkable problem $\Pi=(\Sigma,\mcN_{\Pi},\mcE_{\Pi})$ where
\begin{itemize}
    \item $\Sigma=\{(a,b)\mid a,b\in\mathbb{Z}_{> 0}\}\cup\{D\} $,
    \item $\mcN_{\Pi}=\{\mcN_{\Pi}^0,\mcN_{\Pi}^1,\mcN_{\Pi}^2,\dots\}$ where $\mcN_{\Pi}^i=\{\{(a_1,b_1),\dots,(a_p,b_p),D,D,\dots,D\}\mid a_1,b_1,\dots,a_p,b_p\in \mathbb{Z}_{> 0}, \text{ and for all } k,l,m\leq p \text{ with } l\neq m, \text{ we have } a_k\leq p,\text{ and } b_l\neq b_m\}$, and
    \item $\mcE_{\Pi}=\{\mcE_{\Pi}^0,\mcE_{\Pi}^1,\mcE_{\Pi}^2\}$ where 
    $\mcE_{\Pi}^0=\{\emptyset\}$, $\mcE_{\Pi}^1=\{\{D\}\}$, and $\mcE_{\Pi}^2=\{\{(a_1,b),(a_2,b)\}\mid a_1,a_2,b\in \mathbb{Z}_{> 0} \text{ and }a_1+a_2\geq b+1\}$.
\end{itemize}

It is straightforward to transform any valid solution to $\Pi$ on semi-graphs into a valid solution to $(\text{edge-degree} + 1)$-edge coloring (as defined outside of the context of node-edge-checkable problems) on graphs in $1$ round: simply assign to each edge $e$ of the graph the color $b$ that is the second entry of each of the two pairs that are outputted on the two half-edges incident to $e$ in an assumed correct solution to $\Pi$.
The constraints of $\Pi$ then immediately imply that the obtained edge coloring is proper and does not use a color that is greater than $\edgedeg(e) + 1$ for $e$.

Vice versa, given a semi-graph $S$ with underlying graph $G$, a valid solution for $(\text{edge-degree} + 1)$-edge coloring on $G$ can be straightforwardly transformed into a valid solution for $\Pi$ on $S$ in $1$ round: on each half-edge incident to an edge of rank $1$, simply output $D$; on the two half-edges incident to an edge $e = \{ u, v \}$ of rank $2$, output labels $(a_1,b)$ and $(a_2,b)$ (on half-edges $\{u, e\}$ and $\{v, e\}$, respectively) such that $b$ is the color of edge $e$ in the assumed solution on $G$, $a_1 + a_2 \geq b + 1$, $a_1 \leq \deg_G(u)$, and $a_2 \leq \deg_G(v)$.
Again, it follows directly from the definitions of the two problems that such an output label assignment is possible and must produce a valid solution.

Next, we consider the node-list variant $\Pi^*=\{\Sigma,\mcL_{\Pi},\mcE_{\Pi}\}$ of problem $\Pi$. By definition of $\Pi^{*}$, $\mcL_{\Pi}=\{\mcL_{\Pi}^1,\mcL_{\Pi}^2,\dots\}$ where $\mcL_{\Pi}^i=\{\mcN_{\Pi,\mset}^i \mid \mset\in \mcN_{\Pi}^{j}\text{ for some }j\in\mathbb{Z}\} $. We have that for $\mset=\{(x_1,y_1),\dots,(x_q,y_q),D\dots,D\}\in\mcN_{\Pi}^j$, 
\begin{equation*}
    \mcN_{\Pi,\mset}^i=\{\{(a_1,b_1),\dots,(a_p,b_p),D,\dots D\}\mid \{(a_1,b_1),\dots,(a_p,b_p),(x_1,y_1),\dots,(x_q,y_q),D\dots,D\}\in \mcN_{\Pi}^{i+j}\}.
\end{equation*}

For a label $(a,b)$ assigned to a half-edge $h$, we call $a$ the \emph{degree part} of the label and $b$ the \emph{color part} of the label. On a graph $G$, the constraints of problem $\Pi$ translate to the statement that for an edge $e$ incident on half-edges $h_1$ and $h_2$, the sum of the degree parts of $h_1$ and $h_2$ is at least $1$ greater than the color part which is common for $h_1$ and $h_2$.
We show the following lemma.

\begin{lemma}\label{lem:solveedge}
    For any given instance $(G,\hin)$, $\Pi^{*}$ admits a solution.
\end{lemma}
\begin{proof}
    Let $(e_1,e_2,\dots,e_l)$ and $(f_1,f_2,\dots,f_m)$ be orderings of rank-$2$ edges and rank-$1$ edges respectively. We define a \emph{labeling process} that goes through edges sequentially and assign labels to half-edges incident on them.
    \begin{itemize}
        \item Go through edges $e_1,e_2,\dots, e_l$ sequentially. For an edge $e=e_{t},t\leq l$, do the following. Let $v_1$ and $v_2$ be the nodes incident on $e$. Suppose that degree of $v_1$ and $v_2$ are $i_1$ and $i_2$ respectively. Suppose $\hin(v_1)=\mcN_{\Pi,\mset_1}^{i_1}$ and $\hin(v_2)=\mcN_{\Pi,\mset_2}^{i_2}$ where $\mset_1\in \mcN_{\Pi}^{j_1}$ and $\mcN_{\Pi}^{j_2}$. Let $\mset_1=\{(x_1,y_1),\dots,(x_{b_1},y_{b_1}),D,\dots,D\}$ and $\mset_2=\{(w_1,z_1),\dots,(w_{b_2},z_{b_2}),D,\dots ,D\}$. Let $\chi_1=\{(r_1,s_1),\dots,(r_{d_1},s_{d_1})\}$ (for $d_1\leq i_1$) and $\chi_2=\{(p_1,q_1),\dots,(p_{d_2},q_{d_2})\}$ (for $d_2\leq i_2$) be the multisets of labels already assigned to incident half-edges on $v_1$ and $v_2$ respectively. Note $D\notin \chi_1$ and $D\notin \chi_2$ since only edges of rank $1$ have their incident half-edge labeled $D$ and none of them have chosen a label for their incident half-edges yet. Now,
        \begin{enumerate}[label=(\roman*)]
            \item Choose $c\in\{1,2,\dots,b_1+b_2+d_1+d_2+1\}$ such that $c\neq t$ for
            \[ t\in\{\ y_1,\dots,y_{b_1},z_1,\dots,z_{b_2},s_1,\dots,s_{d_2},q_1,\dots,q_{d_2}\}.\] 
            Note that this is always possible.
            \item Assign to $\{v_1,e\}$ and $\{v_2,e\}$, $(b_1+d_1+1,c)$ and $(b_2+d_2+1,c)$ respectively. Since $b_1+b_2+d_1+d_2\geq c+1$, this ensures edge configurations on rank $2$ edges are always valid.
        \end{enumerate}
        \item Go through $f_1,f_2\dots,f_m$ sequentially. For an edge $e=f_t,t\leq m$, let $\{v,e\}$ be the incident half-edge. Assign to $\{v,e\}$ the label $D$. This ensures that edge configurations on rank $1$ edges are always valid.
        
    \end{itemize}
   
   We now show that any node has a valid node configuration as well. Let $v$ be a node of degree $i$. Let $\hin(v)=\mcN_{\Pi,\mset}^{i}$ where $\mset=\{(a_1,b_1),\dots,(a_c,b_c),D,\dots,D)\in\mcN_{\Pi}^j\}$. Let the half-edges incident on $v$ be $h_1,\dots,h_i$ and assume that they were labeled in the labeling process in the order $h_1,h_2\dots,h_i$. Let their labels be $(x_1,y_1),\dots,(x_d,y_d),D,D\dots,D$ respectively. Such an assumption is valid since half-edges incident to edges of rank $1$ are labeled after half-edges incident to edges of rank $2$ and only half-edges incident to edges of rank $1$ have $D$ as their label. Now, consider the label $(x_t,y_t)$ assigned to $h_t$ for any $t\leq d$. By definition, the labeling process defines the degree part $x_t=c+(t-1)+1=c+t\leq c+d$, and the color part $y_t$ to be different from $y_1,\dots,y_{t-1},b_1,\dots, b_c$. Since $\mset\in\mcN_{\Pi}^j$, we have that $a_1,\dots,a_c\leq c$. Hence the cardinality $i+j$ multiset $\{(a_1,b_1),(a_2,b_2)\dots,(a_c,b_c),(x_1,y_1),(x_2,y_2)\dots,(x_d,y_d),D,D,\dots,D\}$ has the property that any two of its elements other than $D$ have their degree part at most $c+d$ and their color part different. Hence node $v$ has a valid configuration.
\end{proof}

Now combining the equivalence of $\Pi$ and $(\text{edge-degree} + 1)$-edge coloring shown above with \Cref{lem:solveedge}, \Cref{thm:edgesveryformal}, and the fact that \cite[Theorem D.4, arXiv version]{balliu2022edge} provides an $O(\log^{12}\Delta + \log^* n)$-round algorithm for $(\text{edge-degree} + 1)$-edge coloring, we obtain the desired result (where we use that the expression $\frac{\rho}{\rho-\log_{g(n)}a}$ from \Cref{thm:edgesveryformal} is upper bounded by a constant for $\rho=2$ and $a\leq g(n)$).

\edgecolthm*

\subsection{Maximal Matching}\label{sec:applicationmatching}
We define the problem $\Pi=(\Sigma,\mcN_{\Pi},\mcE_{\Pi})$ where 
\begin{itemize}
    \item $\Sigma=\{M,P,O,D\}$
    \item $\mcN_{\Pi}=\{\mcN_{\Pi}^0,\mcN_{\Pi}^1,\dots\}$ where $\mcN_{\Pi}^i=\{\{\chi_1,\chi_2,\dots,\chi_i\}\mid \chi_1,\chi_2,\dots,\chi_i\in\Sigma \text{ and either } (i)\ \chi_j=M \text{ for one }j\leq i, \chi_k\in\{P,O,D\} \text{ for } k\neq j, \text{ or } (ii)\ \chi_i\in\{O,D\} \text{ for all }k\leq i \}$.
    \item $\mcE_{\Pi}^{0}=\{\emptyset\}$, $\mcE_{\Pi}^1=\{\{D\}\}$, and $\mcE_{\Pi}^{2}=\{\{P,O\},\{M,M\},\{P,P\}\}$.
\end{itemize}

It is straightforward to transform any valid solution to $\Pi$ on semi-graphs into a valid solution to maximal matching on graphs in $1$ round: simply put an edge of the input graph into the matching if and only if the two labels assigned to the corresponding edge in the semi-graph on which a valid solution to $\Pi$ is assumed to be given are both $M$.
The constraints of $\Pi$ directly imply that the obtained solution is indeed a valid solution for maximal matching.

Vice versa, given a semi-graph $S$ with underlying graph $G$, a valid solution to maximal matching can be straightforwardly transformed into a valid solution for $\Pi$ on $S$ in $1$ round: on each half-edge incident to an edge of rank $1$, simply output $D$; for each edge $e$ of rank $2$ output $M$ on both incident half-edges if $e$ is matched in the given solution, and otherwise output
\begin{itemize}
    \item $O$ on each half-edge incident on $e$ for which the incident node is not matched in the given solution, and
    \item $P$ on each half-edge incident on $e$ for which the incident node is matched in the given solution.
\end{itemize}
Again, the definitions of $\Pi$ and maximal matching immediately imply that the obtained solution is indeed a valid solution on $S$.

Next, we consider the node-list variant $\Pi^*=\{\Sigma,\mcL_{\Pi},\mcE_{\Pi}\}$ of problem $\Pi$. By definition of $\Pi^{*}$, $\mcL_{\Pi}=\{\mcL_{\Pi}^1,\mcL_{\Pi}^2,\dots\}$ where $\mcL_{\Pi}^i=\{\mcN_{\Pi,\mset}^i:\mset\in \mcN_{\Pi}^{j}\text{ for some }j\in\mathbb{Z}\} $. We have that for $\mset=\{\mset_1,\mset_2,\dots,\mset_j\}\in\mcN_{\Pi}^j$,
\begin{equation*}
    \mcN_{\Pi,\mset}^{i}=\{\{\chi_1,\dots,\chi_i\}\mid\{\chi_1,\dots,\chi_i,\mset_1,\dots,\mset_j\}\in\mcN_{\Pi}^{i+j}\}.
\end{equation*}
We now claim the following.
\begin{lemma}\label{lem:solvematch}
    For any valid input instance $(G,\hin)$, $\Pi^{*}$ admits a valid solution.
\end{lemma}

\begin{proof}
    Let $(e_1,e_2,\dots,e_l)$ and $(f_1,f_2,\dots,f_m)$ be orderings of rank $2$ edges and rank $1$ edges respectively. We define a \emph{labeling process} that goes through edges sequentially and fix the labels of half-edges incident on them.
    \begin{itemize}
        \item Go through edges $e_1,e_2,\dots,e_l$ sequentially. For an edge $e=e_t,t\leq i$, do the following. Let $v_1$ and $v_2$ be the nodes incident on $e$. Suppose $\hin(v_1)=\mcN_{\Pi,\mset_1}^{i_1}$ and $\hin(v_2)=\mcN_{\Pi,\mset_2}^{i_2}$ where $\mset_1\in \mcN_{\Pi}^{j_1}$ and $\mcN_{\Pi}^{j_2}$. Let $\mset_1=\{x_1\dots,x_{b_1},D,\dots,D\}$ and $\mset_2=\{w_1\dots,w_{b_2},D,\dots ,D\}$. Let $\chi_1=\{r_1\dots,r_{d_1}\}$ (for $d_1\leq i_1$) and $\chi_2=\{p_1,\dots,p_{d_2}\}$ (for $d_2\leq i_2$)be the multisets of labels already assigned to incident half-edges on $v_1$ and $v_2$ respectively. Note $D\notin \chi_1$ and $D\notin \chi_2$ since only edges of rank $1$ have their incident half-edge labeled $D$ and none of them have chosen a label for their incident half-edges. Now
        \begin{itemize}
            \item \emph{Case 1}: If $M\notin\{x_1,\dots ,x_{b_1},w_1,\dots,w_{b_2},r_1,\dots,r_{d_1},p_1,\dots,p_{d_2}\}$, then assign label $M$ to half-edges incident on $e$.
            \item \emph{Case 2}: If $M\notin\{x_1,\dots ,x_{b_1},r_1,\dots,r_{d_1}\}$ and $M\in\{w_1,\dots,w_{b_2},p_1,\dots,p_{d_2}\}$, then assign label $O$ to half-edge $\{v_1,e\}$ and label $P$ to half-edge $\{v_2,e\}$.
            \item \emph{Case 3:} If $M\in\{x_1,\dots ,x_{b_1},r_1,\dots,r_{d_1}\}$ and $M\in\{w_1,\dots,w_{b_2},p_1,\dots,p_{d_2}\}$, assign label $P$ to the half-edges incident on $e$.
        \end{itemize}
        \item Go through $f_1,f_2\dots,f_m$ sequentially. For each edge $e=f_t,t\leq m$, assign to the half-edge incident on $f_t$, the label $D$.
    \end{itemize}
    It is straightforward to verify that the all the edges satisfy the edge constraint by the definition of the process.

    We will now show that any node has a valid node configuration. Let $v$ be a node of degree $i$. Let $\hin(v)=\mcN_{\Pi,\mset}^i$ where $\mset=\{\mset_1,\mset_2,\dots,\mset_c,D,D,\dots,D\}\in\mcN_{\Pi}^{j}$. Let the half edges incident on $v$ be $h_1,h_2,\dots,h_i$ and assume they were labeled in that order itself by the process. Let their labels be $\chi_1,\chi_2,\dots,\chi_d,D,D,\dots,D$ respectively where $\chi_t\neq D$ for $t\leq D$. Such an assumption is valid since half-edges incident on edges of rank $2$ are labeled before half-edges incident on edges of rank $1$ are labeled. Assume that $v$ does not have a valid node configuration.
    Then one of the following cases must occur.
    \begin{itemize}
        \item \emph{Case A:} There exists two elements in the multiset $\{\mset_1,\dots,\mset_c,\chi_1,\dots,\chi_d\}$ that are $M$. Since $\mset\in \mcN_{\Pi}^j$ , at most one of $\mset_1,\dots,\mset_c$ is the label $M$. Hence one of the labels $\chi_1,\dots,\chi_i$ must be the label $M$. However, Case $1$ of the labeling process prevents this.
        \item \emph{Case B:} All of the labels of the multiset $\{\mset_1,\dots,\mset_c,\chi_1,\dots,\chi_d\}$ are not $O$. Consider such a label $\rho$. Assume $\rho=M$, then by argument in Case A, $\rho$ is the only element of the multiset $\{\mset_1,\dots,\mset_c,\chi_1,\dots,\chi_d\}$ that is $M$ and the node configuration of $v$ is valid. Hence, $\rho\neq M$. Then $\rho=P$. If $P\in\mset$, this implies that $M\in \mset$ (since $\mset\in \mcN_{\Pi}^j$) which would imply that the node configuration of $u$ is valid. Hence $P\in\{\chi_1,\dots\chi_d\}$ which implies one of the half-edges incident on $u$ is labeled $P$. Let half-edge $h_x$ be that label. But Cases $2$ and $3$ in the labeling process ensure that $h_x$ is labeled $P$ if $M\in\{\mset_1,\dots,\mset_c,\chi_1,\dots,\chi_d\}$. This again implies that the node configuration of $u$ is valid, yielding a contradiction.
    \end{itemize}
\end{proof}

Now combining the equivalence of $\Pi$ and maximal matching shown above with \Cref{lem:solvematch}, \Cref{thm:edgesveryformal}, and the fact that \cite{panconesi2001some} provides an $O(\Delta + \log^* n)$-round algorithm for maximal matching, we obtain the desired result that maximal matching can be solved on trees in $O(\log n/\log \log n)$ rounds (using a simplification analogous to the one in the case of $(\text{edge-degree} + 1)$-edge coloring).

We remark that with a similar approach as for $(\text{edge-degree} + 1)$-edge coloring and maximal matching, it is also straightforward to phrase the other problems mentioned throughout the paper (such as MIS, $(\deg + 1)$-coloring, $(\Delta + 1)$-coloring and $(2 \Delta - 1)$-edge coloring) in the described node-edge-checkability formalism.
In fact, as round elimination also requires problems to be phrased in node-edge-checkable form---though without the minor additional technicality that comes due to semi-graphs---, a description in this form on semi-graphs can be usually derived from the respective round elimination considerations in the literature.

\section*{Acknowledgements}
Funded by the European Union. Views and opinions expressed are however those of the author(s) only
and do not necessarily reflect those of the European Union or the European Research Council. Neither the European Union nor the granting authority can be held responsible for them. This work is supported by ERC grant OLA-TOPSENS (grant agreement number 101162747) under the Horizon Europe funding programme.

\urlstyle{same}
\bibliographystyle{alpha}
\bibliography{references}

\end{document}